\documentclass[letterpaper, 10pt, conference]{ieeeconf}  
\IEEEoverridecommandlockouts                              
\usepackage[applemac]{inputenc}
\usepackage{graphics} 
\usepackage{epsfig} 
\usepackage{epstopdf}
\usepackage{amsmath} 
\usepackage{amssymb}  
\usepackage{epstopdf}
\usepackage{graphicx,xcolor}
\usepackage{caption}
\usepackage{subcaption}
\usepackage{courier}
\usepackage{mathrsfs}
\usepackage{lipsum} 
\usepackage{tikz}
\usetikzlibrary{automata,positioning,shapes,arrows}

\newtheorem{thm}{Theorem}

\newtheorem{cor}{Corollary}

\newtheorem{prop}{Proposition}
\newcommand{\nt}{{\mathbb N}}

\title{\LARGE \bf
Privacy Verification in POMDPs via Barrier Certificates
}

\author{Mohamadreza Ahmadi$^{*}$ \and   Bo Wu$^{*}$ \and Hai Lin \and Ufuk Topcu
\thanks{ 
M. Ahmadi and U. Topcu are with the Department of Aerospace Engineering and Engineering Mechanics, and the Institute for Computational Engineering and Sciences (ICES), University of Texas, Austin, 201 E 24th St, Austin, TX 78712. B. Wu and H. Lin are with the Department of Electrical Engineering, University of Notre Dame, Notre Dame, IN 46556, USA. e-mail: (\{mrahmadi, utopcu\}@utexas.edu,\{bwu3, hlin1\}@nd.edu). 
}
\thanks{$^{*}$ M. Ahmadi and B. Wu contributed equally to this work.}
}

\begin{document}

\maketitle
\thispagestyle{empty}
\pagestyle{empty}

\begin{abstract}

Privacy is an increasing concern in cyber-physical systems that operates over a shared network. In this paper, we propose a method for privacy verification of cyber-physical systems modeled by Markov decision processes (MDPs) and partially-observable Markov decision processes (POMDPs) based on barrier certificates. To this end, we consider an opacity-based notion of privacy, which is characterized by the beliefs in system states. We show that the belief update equations can be represented as discrete-time switched systems, for which we propose a set of conditions for privacy verification  in terms of barrier certificates. We further demonstrate that, for MDPs and for POMDPs, privacy verification can be computationally implemented  by solving a set of semi-definite programs and sum-of-squares programs, respectively. The method is illustrated by an application to privacy verification of an inventory management system. 


\end{abstract}

\section{INTRODUCTION}

Privacy is becoming a rising concern in many modern engineering systems which are  increasingly connected over shared infrastructures, such as power grids \cite{mcdaniel2009security}, healthcare systems \cite{sahi2018privacy}, smart home \cite{courtney2008privacy}, transportation systems \cite{hubaux2004security}, and etc. Potentially malicious intruders may have access to the information available publicly or privately based on which they attempt to infer some ``secret'' associated with the system, such as personal activity preferences, health conditions, and  bank account details. If the privacy is compromised, it could lead to  substantial social or economic loss. Therefore, it is of fundamental importance to design  cyber-physical systems that are provably safe against privacy breaches. 

In recent years, a privacy notion called ``opacity" has received significant attention.  Generally speaking, opacity is a confidentiality property that characterizes a system's capability to conceal its ``secret" information from being inferred by outside observers. These observers are assumed to have full knowledge of the system model, often as a finite automaton, and can observe or partially observe the behaviors of the system, such as the actions performed, but not the states of the system  directly. Various notions of opacity, depending on whether the secret is the behavior of the system in regular languages, initial states, or the current states, have been proposed~\cite{wu2013comparative} and their verification and enforcement are studied in deterministic and probabilistic systems~\cite{jacob2016overview}.

Most existing results on opacity only consider the absolute certainty of  the occurrence of the secret as the privacy violation. However, in practice, in many (partially observable) probabilistic systems, the intruder may only maintain a belief over the system secrets through Bayesian inference, which may still pose a security threat if the intruder has a high confidence that a secret has been observed. Hence,  a new opacity notion was introduced in~\cite{wu2018privacy} for Markov decision processes (MDPs), where the system is considered opaque, if the intruder's confidence that the current state is a secret state never exceeds a given threshold.

In this paper, in addition to studying privacy verification in MDPs, we study  partially observable MDP (POMDP) models with the privacy metric based on opacity. POMDPs  generalize   MDPs with partial observability and are popular in sequential decision-making \cite{cassandra}. Existing studies on POMDPs mostly consider model checking against a given specification~\cite{chatterjee2016decidable}, or policy synthesis to optimize a given performance metric~\cite{junges2017permissive}. Privacy issues in POMDP planning have gained interest only recently. For example, in \cite{yao2015privacy}, privacy is quantified as the average conditional entropy  to be minimized while optimizing the task-related reward in a power grid. An accumulated discounted minimal Bayesian risk was defined in \cite{li2016privacy} as the privacy breach metric to be optimized. Like these two papers, most existing work focuses on privacy measures that are averaged over time. However, minimizing a time average may not be sufficient in some circumstances, because it does not guarantee that the intruder will not have a fairly high confidence about the secret at \emph{some time instant}. In contrast, our notion of privacy is supposed to be satisfied \emph{at any time}.


A key observation we use in \cite{wu2018privacy} is that the intruder's belief update dynamics can be characterized as an autonomous discrete-time switched system whose switching signals are the observed actions. Then, the privacy verification problem can be equivalently cast into verifying whether the  solutions of the belief switched system avoid a privacy unsafe  subset of the belief space, where the privacy specification is violated. 

Safety verification is a familiar subject to the control community~\cite{Guguen200925,1215682,Prajna2006117,7402508,AHMADI201733}. One of the methods for safety verification relies on the construction of a function of the states, called the \emph{barrier certificate} that satisfies a Lyapunov-like inequality~\cite{Prajna2006117}. The barrier certificates have shown to be useful in several system analysis and control problems running the gamut of bounding moment functional of stochastic systems~\cite{ahmadi2016optimization} to collision avoidance of multi-robot systems~\cite{7857061}. It was also shown in~\cite{7236867} that any safe dynamical system admits a barrier certificate. 

In this paper, we propose conditions for  privacy verification  of MDPs and POMDPs using barrier certificates. From a computational stand point, we formulate a set of semi-definite programs (SDPs) and sum-of-squares programs (SOSP) to verify the privacy requirement of MDPs and POMDPs, respectively. We apply the proposed method to a case study of privacy verification of an inventory management system. 

The rest of this paper is organized as follows. In the subsequent section, we present some definitions related to MDPs, POMDPs, belief dynamics and privacy. In Section~\ref{sec:privacy-verification-barrier}, we propose a set of conditions for privacy verification of belief equations represented as discrete-time switched systems based on barrier certificates. In Section~\ref{sec:privacy-mdps-pomdps}, we apply the method based on barrier certificates to the privacy verification problem of MDPs and POMDPs, and present a set of SDP and SOSP sufficient conditions, respectively. In Section~\ref{sec:privacy-example}, we elucidate the proposed privacy verification methodology with an inventory management example. Finally, in Section~\ref{sec:conclusions}, we conclude the paper and give directions for future research.

\textbf{Notation:}
The notations employed in this paper are relatively straightforward. $\mathbb{R}_{\ge 0}$ denotes the set $[0,\infty)$ and $\mathbb{Z}_{\ge 0}$ denotes the set  of integers $\{0,1,2,\ldots\}$. For a finite set $A$, we denote by $|A|$ the cardinality of the set $A$. Given a matrix $Q$, we denote by $Q^T$ the transpose of $Q$. The notation $0_{n \times m}$ is the $n \times m$ matrix with zero entries. For two vectors, $a$ and $b$ with the same size, $a \succeq b$ implies entry-wise inequality. $\mathcal{R}[x]$ accounts for the set of polynomial functions with real coefficients in $x \in \mathbb{R}^n$, $p: \mathbb{R}^n \to \mathbb{R}$ and $\Sigma \subset\mathcal{R}$ is the subset of polynomials with an SOS decomposition; i.e., $p \in  \Sigma[x]$ if and only if there are $p_i \in \mathcal{R}[x],~i \in \{1, \ldots ,k\}$ such that $p = p_i^2 + \cdots +p_k^2$.  


\section{Preliminaries}
\subsection{MDP}
MDPs \cite{puterman2014markov} are  decision-making modeling framework in which the actions have stochastic outcomes. An MDP $\mathcal{M}=(Q,\pi,A,T)$ has the following components:
	\begin{itemize}
		\item $Q$ is a finite set of states with indices $\{1,2,\ldots,n\}$.
		\item $\pi:Q\rightarrow[0,1]$ defines the distribution of the initial states, i.e., $\pi(q)$ denotes the probability of starting at $q\in Q$.
		\item $A$ is a finite set of actions.
		\item $T:Q\times A\times Q\rightarrow [0,1]$ is the probabilistic transition function where 
		\begin{multline}
		T(q,a,q'):=P(q_t=q'|q_{t-1}=q,a_{t-1}=a), \\
		\forall t\in\mathbb{Z}_{\ge 1}, q,q'\in Q, a\in A. \nonumber
		\end{multline}		 
	\end{itemize}

\subsection{POMDP}
POMDPs provide a more general mathematical framework to consider not only the stochastic outcomes of actions, but also the imperfect state observations~\cite{shani2013survey}. Formally,  a POMDP $\mathcal{P}=(Q,\pi,A,T,Z,O)$ is defined with the following components:

	\begin{itemize}
		\item $Q,\pi,A,T$ are the same as the definition of an MDP.
		\item $Z$ is the set of all possible observations. Usually $z\in Z$ is an incomplete projection of the world state $q$, contaminated by sensor noise.
		\item $O:Q\times A \times Z\rightarrow [0,1]$ is the observation function where
		\begin{multline}
		O(q,a,z):=P(z_t=z|q_{t}=q,a_{t-1}=a), \\
	    \forall t\in\mathbb{Z}_{\ge 1}, q\in Q, a\in A, z\in Z. \nonumber
		\end{multline}			 
	\end{itemize}
Furthermore, we assume that there is a set of secret states $Q_s \subset Q$ and we would like to conceal the information that the system is currently in some secret state $q\in Q_s$.

\subsection{Belief Update Equations as  Discrete-Time Switched~Systems}
In \cite{wu2018privacy}, we considered the case that given a system modeled as an MDP $\mathcal{M}$, there is an intruder with potentially malicious intention  that can observe the actions executed but not the states of the system, and tries to determine whether the system  is currently in some secret state with a high confidence. If all the actions are available at every state, then from the intruder's point of view, the system is actually a POMDP with trivially the same observation for every state (since it cannot observe the states at all). In this case, the intruder may  maintain a belief $b_{t-1}:Q\rightarrow [0,1],\sum_{q\in Q}b_t(q)=1$ over $Q$ at time $t-1$. The belief at $t=0$ is defined as $b_0(q)=\pi(q)$ and $b_t(q)$ denotes the probability of system being in state $q$ at time $t$. At time $t+1$, when action $a\in A$ is observed, the belief update can be described as
\begin{equation}\label{equation:MDP belief update}
b_{t}(q') = \sum_{q\in Q}P(q,a,q')b_{t-1}(q).
\end{equation}

We also consider systems modeled as  a POMDP, where we assume that the intruder may have access to the observations in addition to  the executed actions. Therefore, the intruder has to consider a complete history of the past actions and observations to update its belief
%
with Bayes rule: 
\begingroup\makeatletter\def\f@size{9}\check@mathfonts
\begin{align} \label{equation:belief update}
\begin{split}
&b_t(q')
=P(q'|z_t,a_{t-1},b_{t-1})\\
&= \frac{P(z_t|q',a_{t-1},b_{t-1})P(q'|a_{t-1},b_{t-1})}{P(z_t|a_{t-1},b_{t-1})}\\
&= \frac{P(z_t|q',a_{t-1},b_{t-1})\sum_{q\in Q}P(q'|a_{t-1},b_{t-1},q)P(q|a_{t-1},b_{t-1})}{P(z_t|a_{t-1},b_{t-1})} \\
&=\frac{O(q',a_{t-1},z_{t})\sum_{q\in Q}T(q,a_{t-1},q')b_{t-1}(q)}{\sum_{q'\in Q}O(q',a_{t-1},z_{t})\sum_{q\in Q}T(q,a_{t-1},q')b_{t-1}(q)}.
\end{split}
\end{align}
\endgroup


\subsection{Privacy in Belief Space}


Our notion of privacy is defined on the belief space of the intruder, where we require that the intruder, even with  access to the actions and observations since $t=0$, is never confident that the system is in a secret state with a probability larger than or equal to a constant $\lambda \in [0,1]$, at any time $t$:
\begin{equation}\label{equation:privacy requirement}
\sum_{q\in Q_s}b_t(q)\leq\lambda,\forall t.
\end{equation}

The notion of privacy used in this paper is closely related to the current-state opacity (CSO) in discrete event systems~\cite{jacob2016overview}. The CSO definition provides a deterministic notion of privacy in that privacy is breached when the intruder is absolutely sure that the system is currently in a secret state. On the other hand,  in our formulation, the privacy requirement is violated when the intruder is confident with a probability over some threshold.

\section{Privacy Verification Using \\Barrier Certificates}\label{sec:privacy-verification-barrier}

The belief update equations for MDPs~\eqref{equation:MDP belief update} and POMDPs~\eqref{equation:belief update} are discrete-time switched system where the actions $a \in A$ define the switching modes. In the sequel, we develop a technique based on barrier certificates for privacy verification of belief update equations  represented as  discrete-time switched systems.


Consider the following  belief dynamics  
\begin{equation}\label{equation:discretesystem}
b_t = f_a(b_{t-1}),
\end{equation}
where $b$ denote the belief vector belonging to the belief space hyper-cube $[0,1]^{|Q|}$, $a \in A$ is the action that can be interpreted as the switching mode index, $t \in \mathbb{Z}_{\ge 1}$ denote the discrete time instances, the vector fields $\{f_{a}\}_{a \in A}$ with $f_a: [0,1]^{|Q|} \to [0,1]^{|Q|} $, and $b_0 \in \mathcal{B}_0 \subset [0,1]^{|Q|}$ representing the set of initial beliefs. We also define a privacy unsafe set $\mathcal{B}_u \subset [0,1]^{|Q|}$, where the privacy requirement is violated. Verifying whether all the belief evolutions of~\eqref{equation:discretesystem} starting at $\mathcal{B}_0$ avoid a given privacy unsafe set $\mathcal{B}_u$ at a pre-specified time $T$ or for all time is a cumbersome task in general and requires simulating~\eqref{equation:discretesystem} for all elements of the set $\mathcal{B}_0$ and for different sequences of $a \in A$. Furthermore, POMDPs are often computationally intractable to solve exactly~\cite{Hauskrecht2000}. To surmount these challenges, we demonstrate that we can find a barrier certificate which verifies that a given privacy requirement is not violated without the need to solve the belief update equations or the POMDPs directly. 

\begin{thm}\label{theorem-barrier-discrete}
Consider the belief update equation~\eqref{equation:discretesystem}. Given a set of initial beliefs $\mathcal{B}_0 \subset [0,1]^{|Q|}$, an unsafe set $\mathcal{B}_u \subset [0,1]^{|Q|}$ ($\mathcal{B}_0 \cap \mathcal{B}_u = \emptyset$), and a constant $T$, if there exists a function $B:\mathbb{Z} \times [0,1]^{|Q|} \to \mathbb{R}$ such that
\begin{equation}\label{equation:barrier-condition1}
B(T,b_T)  > 0, \quad \forall b_T \in \mathcal{B}_u,
\end{equation}
\begin{equation}\label{equation:barrier-condition11}
 B(0,b_0) < 0, \quad \forall b_0 \in \mathcal{B}_0,
\end{equation}
and
\begin{multline}\label{equation:barrier-condition2}
B(t,f_a(b_{t-1})) - B(t-1,b_{t-1}) \le 0, \\ \forall t \in \{1,2,\ldots,T\},~\forall a \in A,
\end{multline}
then there exist no solution of belief update equation~\eqref{equation:discretesystem} such that $b_0 \in \mathcal{B}_0$, and $b_T \in \mathcal{B}_u$ for all $a \in A$.
\end{thm}
\begin{proof}
The proof is carried out by contradiction. Assume at time instance $T$ there exit a solution to \eqref{equation:discretesystem} such that $b_0 \in \mathcal{B}_0$ and $b_T \in \mathcal{B}_u$. From inequality~\eqref{equation:barrier-condition2}, we have 
$$
B(t,b_t) \le B(t-1,b_{t-1})
$$
for all $t\in \{1,2,\ldots,T\}$ and all actions $a \in A$. Hence, $B(t,b_t) \le B(0,b_0)$ for all $t \in \{1,2,\ldots,T\}$. Furthermore, inequality~\eqref{equation:barrier-condition11} implies that 
$$
B(0,b_0) < 0 
$$
for all $b_0 \in \mathcal{B}_0$.  Since the choice of $T$ can be arbitrary, this is a contradiction because it implies that $B(T,b_T) \le B(0,b_0) < 0$. Therefore, there exist no solution of \eqref{equation:discretesystem} such that $b_0 \in \mathcal{B}_0$ and $b_T \in \mathcal{B}_u$ for any sequence of actions $a \in A$.
\end{proof}

Theorem~\ref{theorem-barrier-discrete} checks whether the privacy requirement is not violated at a particular point in time $T$. We can generalize this theorem to the case for verifying privacy for all time. In this case, the barrier certificate is time-invariant. 

\begin{cor}\label{theorem-barrier-discrete-alltime}
Consider the belief switched dynamics~\eqref{equation:discretesystem}. Given a set of initial conditions $\mathcal{B}_0 \subset [0,1]^{|Q|}$, and an unsafe set $\mathcal{B}_u \subset [0,1]^{|Q|}$ ($\mathcal{B}_0 \cap \mathcal{B}_u = \emptyset$), if there exists a function $B: [0,1]^{|Q|} \to \mathbb{R}$ such that
\begin{equation}\label{equation:barrier-condition--1}
B(b)  > 0, \quad \forall b \in \mathcal{B}_u,
\end{equation}
\begin{equation}\label{equation:barrier-condition--11}
 B(b) < 0, \quad \forall b \in \mathcal{B}_0,
\end{equation}
and
\begin{equation}\label{equation:barrier-condition--2}
B\left(f_a(b_{t-1})\right) - B(b_{t-1}) \le 0,     
\end{equation}
then there exist no solution of~\eqref{equation:discretesystem} such that $b_0 \in \mathcal{X}_0$ and $b_t \in \mathcal{X}_u$ for all $t \in \mathbb{Z}_{\ge 1}$ and any sequence of actions $a \in A$. Hence, the privacy requirement is not violated.
\end{cor}

\section{Privacy Verification in MDPs and POMDPs}\label{sec:privacy-mdps-pomdps}

In the previous section, we discussed conditions for privacy verification of general belief update equations using barrier certificates. Next, we show that the barrier certificates can be used for privacy verification of MDPs and POMDPs. To this end, we define the privacy unsafe set $\mathcal{B}_u$ to be the complement of the privacy requirement~\eqref{equation:privacy requirement} inspired by the notion of opacity. That is,
\begin{equation}\label{eq:privacy-unsafe-set}
\mathcal{B}_u = \left \{ b \in \mathbb{R}^{|Q|} \mid \sum_{q\in Q_s} b_t(q)> \lambda  \right\}.
\end{equation}
Hence, given a set of initial beliefs $\mathcal{B}_0$, if there exists a barrier certificate verifying privacy with respect to $\mathcal{B}_u$, then we infer that the privacy requirement is satisfied, i.e., \\$\sum_{q\in Q_s}b_t(q)\le \lambda $.

In the following, we formulate a set of conditions in terms of SDPs or SOSPs (refer to Appendix~\ref{app:SOS} for more details on SOSPs) to verify whether a given MDP or a POMDP, respectively,  satisfies a privacy requirement. 


\subsection{Privacy Verification for MDPs via SDPs}

For MDPs, the belief update equation can be described as  a linear discrete-time switched system 
\begin{equation}\label{equation:MDP belief update2}
b_{t+1}(q') = H_a~b_{t}(q') = \sum_{q\in Q}P(q,a,q')b_t(q),
\end{equation}
where $H_a \in \mathbb{R}^{|Q| \times |Q|}$, $a \in A$. Furthermore, the privacy requirement~\eqref{eq:privacy-unsafe-set} describes a  half-space in the belief space hyper-cube.  Denote by $\bar{b}$ the augmentation of the belief states by $1$, i.e., $\bar{b} = \begin{bmatrix} b^T & 1 \end{bmatrix}^T \in \mathbb{R}^{|Q|+1}$. We define the set of initial beliefs to be a convex polytope represented by the intersection of a set of  half-spaces in the augmented belief space
\begin{equation} \label{eq:initial-linear}
\mathcal{B}_0 = \left\{ b_0 \in \mathbb{R}^{|Q|} \mid \bar{E}_0^T \bar{b}_0 \succeq 0_{n_0} \right\},
\end{equation}
where $\bar{E}_0 \in \mathbb{R}^{n_0 \times (|Q|+1)} $.

The privacy unsafe set  can be re-written, respectively, as
\begin{equation} \label{eq:unsafe-linear}
\mathcal{B}_u = \left\{ b \in \mathbb{R}^{|Q_s|} \mid  \bar{b}^T \bar{W} \bar{b} > 0\right\},
\end{equation}
where $$\bar{W} = \begin{bmatrix} 0_{|Q| \times |Q|} & 0_{1\times 1} \\
 \boldsymbol{w}^T & -\lambda \end{bmatrix},  $$
with $\boldsymbol{w} \in \mathbb{R}^{|Q|}$ and $\boldsymbol{w}(i)= 1$ for $i=q \in Q_s$ and $\boldsymbol{w}(i)= 0$ otherwise.


At this point, we are ready to state the SDP conditions for verifying privacy of a given MDP. 

\begin{cor}\label{cor:LMI-privacy}
Consider the MDP belief update dynamics as given in~\eqref{equation:MDP belief update2}, the unsafe set~\eqref{eq:unsafe-linear}, and the set of initial beliefs~\eqref{eq:initial-linear}. If there exist a matrix $V \in \mathbb{S}^{{|Q|}+1}$, a  matrix with positive entries  $U \in \mathbb{S}^{n_0}$, and a positive constant $s^u$ such that
\begin{equation} \label{eq:LMI1}
V - s^u \bar{W}  >0,
\end{equation}
\begin{equation}\label{eq:LMI2}
-V -  \bar{E}_0 U \bar{E}_0^T >0,
\end{equation}
and 
\begin{equation}\label{eq:LMI3}
H^T_a V H_a - V <0,~~\forall a \in A,
\end{equation}
then the privacy requirement~\eqref{equation:privacy requirement} is satisfied for all time.
\end{cor}
\begin{proof}
We show that each of the SDP conditions of \eqref{eq:LMI1}-\eqref{eq:LMI3} correspond to conditions~\eqref{equation:barrier-condition--1}-\eqref{equation:barrier-condition--2}, respectively, for the barrier certificate
$$
B(\bar{b}) = \bar{b}^T(q)~V~\bar{b}(q).
$$
Multiplying both sides of~\eqref{eq:LMI1} from left and right respectively with $\bar{b}^T(q)$ and $\bar{b}(q)$, respectively, gives
$$
\bar{b}^T(q) V \bar{b}(q) - s^u \bar{b}^T(q) \bar{W}  \bar{b}(q) >0.
$$
Since $s^u>0$, from S-procedure, we conclude that $\bar{b}^T(q) V \bar{b}(q)>0$ only if  $\bar{b}^T(q) \bar{W}   \bar{b}(q)>0$ (because $\bar{b}^T V \bar{b}(q) > s^u \bar{b}^T(q) \bar{W}   \bar{b}(q) $). Moreover,  $\bar{b}^T(q) \bar{W}   \bar{b}(q)>0 $ implies that  $\sum_{q\in Q_s}~b_t(q)> \lambda$. Therefore, condition~\eqref{equation:barrier-condition--1} is satisfied. Similarly, we can show, via S-procedure~\cite{S003614450444614X}, that if the linear matrix inequality~\eqref{eq:LMI2} is satisfied, condition~\eqref{equation:barrier-condition--11} holds. This is due to the fact that the polytope $\mathcal{B}_0$ is contained in the ellipsoid represented by $\bar{b}^T \bar{E}_0 U \bar{E}_0^T \bar{b} > 0$ and the positive entries of $U$ are the S-procedure coefficients based on the construction in~\cite[p. 76]{Johanssonthesis}.

Finally, multiplying both sides of~\eqref{eq:LMI3} from left and right respectively with $\bar{b}^T(q)$ and $\bar{b}(q)$ yields
 $$
\bar{b}^T(q) \left( H^T_a V H_a - V \right) \bar{b}(q) <0,~~\forall a \in A.
 $$
 That is,
  $$
\bar{b}^T(q)  H^T_a V H_a \bar{b}(q) - \bar{b}^T(q) V  \bar{b}(q) <0,~~\forall a \in A,
 $$
 which in turn implies that~\eqref{equation:barrier-condition--2} holds for $B(\bar{b}) = \bar{b}^T(q)~V~\bar{b}(q)$. Therefore, from Corollary~\ref{theorem-barrier-discrete-alltime}, the solutions of the MDP belief update equation~\eqref{equation:MDP belief update2} are safe with respect to the privacy unsafe set~\eqref{eq:unsafe-linear}. Hence,  the privacy requirement is satisfied.
 
\end{proof}


\subsection{Privacy Verification for POMDPs via SOSP}

The belief update equation~\eqref{equation:belief update} for a POMDP is a rational function in the belief states $b_t(q)$, $q \in Q_s$
\begin{multline}\label{eq:belief-update-rational}
b_t(q') = \frac{S_a\left( b_{t-1}(q') \right)}{R_a\left( b_{t-1}(q') \right)} \\
= \frac{O(q',a_{t-1},z_{t})\sum_{q\in Q}T(q,a_{t-1},q')b_{t-1}(q)}{\sum_{q'\in Q}O(q',a_{t-1},z_{t})\sum_{q\in Q}T(q,a_{t-1},q')b_{t-1}(q)}
\end{multline}

 Moreover, the privacy unsafe set~\eqref{eq:privacy-unsafe-set} is a semi-algebraic set, since it can be described by a polynomial inequality. We further assume the set of initial beliefs is also given by a semi-algebraic set
\begin{equation}\label{eq:intitial-belief-semialgebraic}
\mathcal{B}_0 = \bigg \{ b_0 \in \mathbb{R}^{|Q_s|} \mid l_i^0(b_0) \le 0,~i = 1,2,\ldots, n_0 \bigg\}.
\end{equation}
At this stage, we are ready to present conditions based on SOSP to verify privacy of a given POMDP. 

\begin{cor}
Consider the POMDP belief update dynamics~\eqref{eq:belief-update-rational}, the privacy unsafe set~\eqref{eq:privacy-unsafe-set}, the set of initial beliefs \eqref{eq:intitial-belief-semialgebraic}, and a constant $T>0$. If there  exist polynomial functions $B \in \mathcal{R}[t,b]$ with degree $d$,   $p^u \in {\Sigma}[b]$,  $p_i^0 \in {\Sigma}[b]$, $i = 1,2,\ldots, n_0$, and constants $s_1,s_2>0$ such that
\begin{equation}\label{eq:setssos1}
B\left(T,b_T\right) -  p^u(b_T) \left(  \sum_{q\in Q_s}b_T(q) - \lambda \right)- s_1 \in \Sigma \left[b_T\right],
\end{equation}
\begin{equation}\label{eq:setssos2}
-B\left(0,b_0\right) + \sum_{i=1}^{n_0} p_i^0(b_0) l_i^0(b_0) - s_2 \in \Sigma \left[b_0\right],
\end{equation}
and 
\begin{multline}\label{eq:setssos3}
- {R_a\left( b_{t-1} \right)}^d\left(B\left(t,\frac{S_a\left( b_{t-1} \right)}{R_a\left( b_{t-1} \right)} \right) - B(t-1,b_{t-1}) \right)  \\\in \Sigma[t,b_{t-1}],~~\forall  t \in \{1,2,\ldots,T\},
\end{multline}
then the privacy requirement~\eqref{equation:privacy requirement} is satisfied for all $t\in \{1,2,\ldots,T\}$.
\end{cor}
\begin{proof}
SOS conditions~\eqref{eq:setssos1} and~\eqref{eq:setssos2} are a direct application of Propositions~\ref{chesip} and~\ref{spos} in Appendix~A to verify conditions \eqref{equation:barrier-condition1} and \eqref{equation:barrier-condition11}, respectively. Furthermore, condition~\eqref{equation:barrier-condition2} for system \eqref{eq:belief-update-rational} can be re-written as  
$$
B\left(t,\frac{S_a\left( b_{t-1} \right)}{R_a\left( b_{t-1} \right)} \right) - B(t-1,b_{t-1})>0.
$$
Given the fact that $R_a\left( b_{t-1}(q') \right)$ is a positive polynomial of degree one, we can relax the above inequality into an SOS condition given by
\begin{multline}
- {R_a\left( b_{t-1} \right)}^d\left(B\left(t,\frac{S_a\left( b_{t-1} \right)}{R_a\left( b_{t-1} \right)} \right) - B\left(t-1,b_{t-1} \right) \right)  \\\in \Sigma[t,b_{t-1}]. \nonumber
\end{multline}
Hence, if ~\eqref{eq:setssos3} holds, then~\eqref{equation:barrier-condition2}  is satisfied as well. Then, from Theorem~\ref{theorem-barrier-discrete}, we infer that there is no $b_t(q)$ at time $T$ such that $b_0(q) \in \mathcal{B}_0$ and  $\sum_{q\in Q_s}b_T(q)> \lambda$. Equivalently, the privacy requirement is satisfied at time $T$. That is, $\sum_{q\in Q_s}b_T(q)\le \lambda$.
\end{proof}

We can also verify privacy for all time for a given POMDP, which is based on Corollary~\ref{theorem-barrier-discrete-alltime}.

\begin{figure*}[t!]\label{fig:mdp}
	\centering	
	\begin{tikzpicture}[shorten >=1pt,node distance=4cm,on grid,auto, thick,scale=.75, every node/.style={transform shape}]
	\node[state] (q_1)   {$q_1$};
	\node[state] (q_2) [below left = 6cm of q_1] {$q_2$};
	\node[state] (q_3) [below right =6cm of q_1] {$q_3$};
	
	\path[->]
	(q_1) edge [pos=0.5, loop, above=0.1] node  {$\sigma_1,0.15;\sigma_2,0.25$} (q_1)
	(q_1) edge [pos=0.5, bend right, above=0.5,sloped] node {$\sigma_1,0.45;\sigma_2,0.25$} (q_2)
	(q_1) edge [pos=0.5, bend left, above=0.5,sloped] node {$\sigma_1,0.4;\sigma_2,0.5$} (q_3)
	
	(q_2) edge [pos=0.5, loop left, left=0.1] node  {$\sigma_1,0.2;\sigma_2,0.1$} (q_2)
	(q_2) edge [pos=0.5, bend right, above=0.5,sloped] node {$\sigma_1,0.2;\sigma_2,0.35$} (q_1)
	(q_2) edge [pos=0.5, above=0.5] node {$\sigma_1,0.6;\sigma_2,0.55$} (q_3)
	
	(q_3) edge [pos=0.5, loop right, right=0.1] node {$\sigma_1,0.5;\sigma_2,0.4$} (q_3)
	(q_3) edge [pos=0.5, bend left, above=0.5,sloped] node {$\sigma_1,0.3;\sigma_2,0.1$} (q_1)
	(q_3) edge [pos=0.5, bend left, above=0.5,sloped] node {$\sigma_1,0.2;\sigma_2,0.5$} (q_2)
	;
	
	\end{tikzpicture}
%
%
%
%
%
%
%
	
	\caption{The MDP in Example I}

\end{figure*}
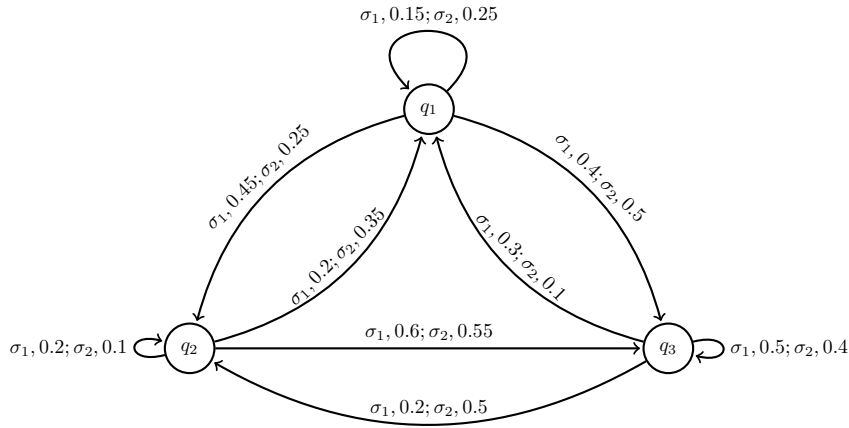

\begin{cor}\label{Cor:4}
Consider the POMDP belief update dynamics~\eqref{eq:belief-update-rational}, the privacy unsafe set~\eqref{eq:privacy-unsafe-set}, and the set of initial beliefs \eqref{eq:intitial-belief-semialgebraic}. If there is exist polynomial functions $B \in \mathcal{R}[b]$ with degree $d$,   $p^u \in {\Sigma}[b]$,  $p_i^0 \in {\Sigma}[b]$, $i = 1,2,\ldots, n_0$, and constants $s_1,s_2>0$ such that
\begin{equation}\label{eq:setssos-1}
B\left(b\right) -  p^u(b) \left(  \sum_{q\in Q_s}b(q) - \lambda \right)- s_1 \in \Sigma \left[b\right],
\end{equation}
\begin{equation}\label{eq:setssos-2}
-B\left(b_0\right) + \sum_{i=1}^{n_0} p_i^0(b_0) l_i^0(b_0) - s_2 \in \Sigma \left[b_0\right],
\end{equation}
and 
\begin{multline}\label{eq:setssos-3}
- {R_a\left( b_{t-1} \right)}^d\left(B\left(\frac{S_a\left( b_{t-1} \right)}{R_a\left( b_{t-1} \right)} \right) - B(b_{t-1}) \right)  \in \Sigma[b_{t-1}],
\end{multline}
then the privacy requirement~\eqref{equation:privacy requirement} is satisfied for all time.
\end{cor}

\section{Numerical Example:\\ Privacy in an Inventory Management System} \label{sec:privacy-example}

In this section, we illustrate the proposed privacy verification method by applying it to an inventory management system. The numerical experiments are carried out on a MacBook Pro 2.9GHz Intel Core i5 and 8GB of RAM. The SDPs are solved using YALMIP~\cite{YALMIP} and the SOSPs are solved using the SOSTOOLs~\cite{PAVPSP13} parser and solvers such as Sedumi~\cite{Stu98}.

\subsection{Example I}
We use the same example from \cite{wu2018privacy}. Suppose the MDP $\mathcal{M}$ has three states $Q=\{q_1,q_2,q_3\}$ representing different inventory levels of a company. The states $q_2,q_3\in Q_s$ correspond to the low and high inventory levels, respectively, and are the secret states. If the intruder, say  a competitor or a supplier, has  information over the current inventory levels being high or low, they may manipulate the price of the goods, and thus negatively affect the company's profit. Therefore, it is of the company's interest to conceal the inventory levels from the potential intruders.  $q_1$ is a non-secret state representing the normal inventory level. $A = \{\sigma_1,\sigma_2\}$ represents two different actions denoting different purchasing quantities. The initial condition is  $\pi(s_1)=0.1,\pi(s_2)=0.2,\pi(s_3)=0.2$. The transition probabilities are as shown in the following matrices for action $\sigma_1$ and $\sigma_2$,  $H_{\sigma_a}(i,j)=T(q_j,\sigma,q_i)$.
\begin{equation}
H_{\sigma_1}=\begin{bmatrix}
   0.15, &0.2 ,&0.3 \\
   0.45,&0.2,&0.2 \\
   0.4,&0.6,&0.5
\end{bmatrix},
H_{\sigma_2}=\begin{bmatrix}
   0.25, &0.35, &0.1 \\
   0.25, &0.1,    &0.5 \\
   0.5, &0.55,  &0.4
\end{bmatrix}    
\end{equation}
The randomness of the inventory level after the purchasing action is due to the random demand levels. The privacy requirement is 
\begin{equation}\label{equation:privacy requirement example}
    b_t(q_2)+b_t(q_3)\leq \gamma,\forall t.
\end{equation}
Based on Corollary~\ref{cor:LMI-privacy}, we check whether the above privacy requirement is satisfied for $\gamma=0.85$. The SDPs~\eqref{eq:LMI1} to \eqref{eq:LMI3} are solved certifying that the privacy requirement~\eqref{equation:privacy requirement example} is satisfied, where we found the following barrier certificate (up to 0.01 precision) in 2.5803 seconds
\begin{multline*}
B(\bar{b}) = \\ \begin{bmatrix} b(q_1) \\ b(q_2) \\ b(q_3) \\ 1 \end{bmatrix}^T ~\begin{bmatrix} 2.98 & -0.83 & -0.61 & 0 \\  -0.83 & 0.07 & 3.89 & 0.92 \\  -0.61 & 3.89 & -1.33 & -0.74 \\ 0 & 0.92 & -0.74 & 1.72   \end{bmatrix}~\begin{bmatrix} b(q_1) \\ b(q_2) \\ b(q_3) \\ 1 \end{bmatrix}.
\end{multline*}
Therefore, the high and low inventory levels are private. Furthermore, in order to find the best achievable privacy requirement, we decrease $\gamma$ and search for a barrier certificate based on Corollary~\ref{cor:LMI-privacy}. We could find the smallest value for $\gamma^* = 0.42$ below which no certificate for privacy verification could  be found. 

\subsection{Example II}
Following our MDP example, besides the purchasing action, the intruder may also have access to the intervals between the two consecutive purchases, which suggests a POMDP $\mathcal{P}$ model that has the same state space $Q$, initial condition $\pi$, action set $A$, transition probabilities $T$. Additionally, $\mathcal{P}$ has the observation set $Z=\{z_0,z_1\}$ which represents a short and a long purchasing intervals respectively. The observation function is defined as below where  $O_{\sigma}(i,j)=O(q_i,\sigma,z_j)$
\begin{equation}\label{equation:observation probability}
O_{\sigma_1}=\begin{bmatrix}
   0.7, &0.3 \\
   0.5,&0.5 \\
   0.8,&0.2
\end{bmatrix},
O_{\sigma_2}=\begin{bmatrix}
   0.8, &0.2 \\
   0.6,&0.4 \\
   0.2,&0.8
\end{bmatrix}.
\end{equation}
The privacy requirement is  (\ref{equation:privacy requirement example}) with $\gamma =0.42$ to make sure that the inventory level being too high or too low is not disclosed with confidence larger than $0.42$. We check the SOSPs~\eqref{eq:setssos-1} to \eqref{eq:setssos-3} where fix the degree $d$ of the barrier certificate. We could not find a certificate for privacy even for $d=10$. In order to find an upper-bound on the achievable privacy requirement, we increase the degree of the barrier certificates from $2$ to $10$ and look for the smallest value of $\gamma$, for which privacy verification could be certified. Table~\ref{t11} outlines the obtained results. As it can be observed from the table, by increasing the degree of the barrier certificate, we can find a tighter upper-bound on the best achievable privacy level. The barrier certificate of degree $2$ (excluding terms smaller than $10^{-4}$) constructed using Corollary~\ref{Cor:4} is provided below
\begin{multline*}
B(b) = 0.1629 b(q_1)^2 - 3.9382b(q_2)^2 + 09280 b(q_3)^2 \\
- 0.0297 b(q_1) b(q_2) - 4.4451 b(q_2) b(q_3) - 0.0027 b(q_1) 
\\- 2.0452 b(q_2) + 9.2633.
\end{multline*}
\begin{table}[!t]
\caption{Numerical results for Example II.}
\label{t11}
\centering
\begin{tabular}{c||c|c|c|c|c}

\bfseries $d$ & 2 & 4 & 6 & 8 & 10   \\
\hline
\bfseries $\gamma^*$ & 0.93 & 0.88 & 0.80 & 0.74 & 0.69    \\
\hline
\bfseries Computation Time (s) & 5.38 & 8.37 & 12.03 & 18.42 & 27.09    \\
\end{tabular}
\end{table}


%
%
%
%
%

\section{CONCLUSIONS AND FUTURE WORK} \label{sec:conclusions}

We proposed a method for privacy verification of MDPs and POMDPs based on barrier certificates. We demonstrated that the privacy verification can be carried out in terms of an SDP problem for MDPs and an SOSP problem for POMDPs. The method was applied to the privacy verification problem of an inventory management system.

The formulation presented here assumes a unified barrier certificate for all actions $a \in A$. A more conservative but more computationally efficient approach to address the privacy verification problem of MDPs and POMDPs is to consider non-smooth barrier certificates, which are composed of a the convex hull, max, or min a set of local barrier certificates for different actions~\cite{7937882,ahmadi2018controller}.  

In addition to privacy verification, the proposed method based on barrier certificates can be used to design a sequence of actions such that  some given privacy requirement is satisfied. To this end, we follow the footsteps of the contributions on synthesizing switching sequences such that some cost is minimized~\cite{7908990}.

%


\bibliography{references}
\bibliographystyle{IEEEtran}


\appendix

\subsection{Sum-of-Squares Polynomials} \label{app:SOS}
A polynomial $p(x)$ is a sum-of-squares polynomial if $\exists p_i(x) \in \mathcal{R}[x]$, $i \in \{1, \ldots, n_d\}$ such that $p(x) = \sum_i p_i^2(x)$. Hence $p(x)$ is clearly non-negative. A set of polynomials $p_i$ is called \emph{SOS decomposition} of $p(x)$. The converse does not hold in general, that is, there exist non-negative polynomials which do not have an SOS decomposition~\cite{Par00}.  The computation of SOS decompositions, can be cast as an SDP (see~\cite{choi1995sums,Par00,CTVG99}). The Theorem below proves that, in sets satisfying a property stronger than compactness, any positive polynomial can be expressed as a combination of sum-of-squares polynomials and polynomials describing the set.  

For a set of polynomials $\bar{g} = \{g_1(x), \ldots, g_m(x)\}$, $m \in \nt$, the \emph{quadratic module} generated by $m$ is 
\begin{equation}
M(\bar{g}):= \left\lbrace \sigma_0 +\sum_{i = 1}^{m} \sigma_i g_i | \sigma_i \in \Sigma[x]\right\rbrace.
\end{equation}
A quadratic module $M\in \mathcal{R}[x]$ is said \emph{archimedean} if $\exists N \in \nt $ such that $$N - |x|^2 \in M.$$
An archimedian set is always compact~\cite{NS08}. It is the possible to state~\cite[Theorem 2.14]{Las09}
\begin{thm}[Putinar Positivstellensatz]
\label{thm:Psatz}
Suppose the quadratic module $M(\bar{g})$ is archimedian. Then for every $f \in \mathcal{R}[x]$, $$f>0~\forall~x\in \{x | g_1(x)\geq 0, \ldots, g_m(x)\geq 0 \} \Rightarrow f \in (\bar{g}).$$
\end{thm}

The subsequent proposition formalizes the problem of constrained positivity of polynomials which is a direct result of applying Positivstellensatz.
\begin{prop}[\cite{chesi2010lmi}] \label{chesip}
Let $\{a_i\}_{i=1}^k$ and $\{b_i\}_{i=1}^l$ belong to $\mathcal{P}$, then
\begin{eqnarray}
p(x) \ge 0 \quad &\forall x \in \mathbb{R}^n: a_i(x)=0, \, \forall i=1,2,...,k & \nonumber \\
& \text{and} \quad b_j(x) \ge 0, \, \forall j=1,2,...,l&
\end{eqnarray}
is satisfied, if the following holds
\begin{eqnarray} \label{chesieq}
&\exists r_1,r_2,\ldots,r_k \in \mathcal{R}[x] \quad \text{and} \quad \exists s_0,s_1,\ldots,s_l \in \Sigma[x] & \nonumber \\
&p = \sum_{i=1}^k r_i a_i +\sum_{i=1}^l s_i b_i +s_0&
\end{eqnarray}
\end{prop}
\begin{prop} \label{spos}
The multivariable polynomial $p(x)$ is strictly positive ($p(x)>0 \quad \forall x \in \mathbb{R}^n$), if there exists a $\lambda > 0$ such that
\begin{equation}
\big( p(x) - \lambda \big) \in \Sigma[x].
\end{equation}
\end{prop}


\end{document}